\documentclass[a4paper,UKenglish,cleveref, autoref, thm-restate]{lipics-v2021}

\title{Lower Bounds for Induced Cycle Detection \\in Distributed Computing} 

\titlerunning{Lower Bounds for Induced Cycle Detection in Distributed Computing} 

\author{Fran\c{c}ois Le Gall}{Graduate School of Mathematics, Nagoya University, Japan}{legall@math.nagoya-u.ac.jp}{}{}

\author{Masayuki Miyamoto}{Graduate School of Mathematics, Nagoya University, Japan}
{masayuki.miyamoto95@gmail.com}{}{}

\authorrunning{F. Le Gall and M. Miyamoto} 

\Copyright{Fran\c{c}ois Le Gall and Masayuki Miyamoto} 



\category{} 

\relatedversion{} 

\supplement{}

\acknowledgements{The authors are grateful to Keren Censor-Hillel, Orr Fischer, Pierre Fraigniaud, Dean Leitersdorf and Rotem Oshman for helpful discussions and comments, and Shin-ichi Minato for his support. This work was partially supported by JSPS KAKENHI grants Nos.~JP19H04066, JP20H05966, JP20H00579, JP20H04139, JP21H04879 and by the MEXT Quantum Leap Flagship Program (MEXT Q-LEAP) grants No.~JPMXS0118067394 and JPMXS0120319794.}

\nolinenumbers 

\hideLIPIcs  

\EventEditors{John Q. Open and Joan R. Access}
\EventNoEds{2}
\EventLongTitle{42nd Conference on Very Important Topics (CVIT 2016)}
\EventShortTitle{CVIT 2016}
\EventAcronym{CVIT}
\EventYear{2016}
\EventDate{December 24--27, 2016}
\EventLocation{Little Whinging, United Kingdom}
\EventLogo{}
\SeriesVolume{42}

\usepackage{xspace}
\newcommand{\congest}{\ensuremath{\mathsf{CONGEST}}\xspace}
\newcommand{\qcongest}{\ensuremath{\mathsf{QUANTUM~CONGEST}}\xspace}
\newcommand{\clique}{\ensuremath{\mathsf{CONGESTED~CLIQUE}}\xspace}
\newcommand{\local}{\ensuremath{\mathsf{LOCAL}}\xspace}

\EventEditors{Hee-Kap Ahn and Kunihiko Sadakane}
\EventNoEds{2}
\EventLongTitle{32nd International Symposium on Algorithms and Computation (ISAAC 2021)}
\EventShortTitle{ISAAC 2021}
\EventAcronym{ISAAC}
\EventYear{2021}
\EventDate{December 6--8, 2021}
\EventLocation{Fukuoka, Japan}
\EventLogo{}
\SeriesVolume{212}
\ArticleNo{65}

\begin{document}

\maketitle

\begin{abstract}
The distributed subgraph detection asks, for a fixed graph $H$, whether the $n$-node input graph contains $H$ as a subgraph or not. 
In the standard \congest model of distributed computing, the complexity of clique/cycle detection and listing has received a lot of attention recently. 

In this paper we consider the induced variant of subgraph detection, where the goal is to decide whether the $n$-node input graph contains $H$ as an \emph{induced} subgraph or not. We first show a $\tilde{\Omega}(n)$ lower bound for detecting the existence of an induced $k$-cycle for any $k\geq 4$ in the \congest model. This lower bound is tight for $k=4$, and shows that the induced variant of $k$-cycle detection is much harder than the non-induced version. This lower bound is proved via a reduction from two-party communication complexity. We complement this result by showing that for $5\leq k\leq 7$, this $\tilde{\Omega}(n)$ lower bound cannot be improved via the two-party communication framework.

We then show how to prove stronger lower bounds for larger values of $k$. More precisely, we show that detecting an induced $k$-cycle for any $k\geq 8$ requires $\tilde{\Omega}(n^{2-\Theta{(1/k)}})$ rounds in the \congest model, nearly matching the known upper bound $\tilde{O}(n^{2-\Theta{(1/k)}})$ of the general $k$-node subgraph detection (which also applies to the induced version) by Eden, Fiat, Fischer, Kuhn, and Oshman~[DISC 2019].

Finally, we investigate the case where $H$ is the diamond (the diamond is obtained by adding an edge to a 4-cycle, or equivalently removing an edge from a 4-clique), and show non-trivial upper and lower bounds on the complexity of the induced version of diamond detecting and listing.
\end{abstract}

\section{Introduction}
\subparagraph*{Background.}
The subgraph detection problem asks us to decide if the $n$-node input graph contains a copy of some fixed subgraph $H$ or not. This problem has received a lot of attention in the past 40 years, and has recently been investigated in the setting of distributed computing as well.
There are actually two versions for this problem.
The first version simply requires to decide if the input network contains $H$. The second version cares about induced $H$, and asks to decide if the input network contains a \textit{vertex-induced} copy of $H$. We refer to Figure~\ref{fig:subgraph_detection} for an illustration of the difference between the two versions. In this paper we call the former version ``non-induced $H$ detection'' and the latter version ``induced $H$ detection''.  

\begin{figure}[tbp]
\centering
    \includegraphics[width=0.1\hsize]{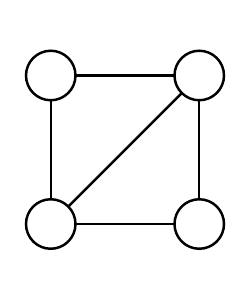}
    \caption{This graph contains a 4-cycle as a subgraph but not an induced 4-cycle.} 
    \label{fig:subgraph_detection}
\end{figure}

When considering subgraph detection in the (synchronous) distributed setting, the communication network is identified with the input graph, i.e., we ask whether the $n$-node communication network contains $H$ as a subgraph (induced or non-induced, depending on the version considered). The complexity is characterized by the number of rounds of (synchronous) communication needed to solve the problem. For networks with unbounded bandwidth  (the so-called \local model in distributed computing), 
both versions of the subgraph detection problem are essentially trivial: for any $O(1)$-node subgraph $H$, the problem can be solved in $O(1)$ rounds by a naive approach. For networks with bounded bandwidth  (the so-called \congest model in distributed computing, in which the size of each message is restricted to $O(\log n)$ bits), on the other hand, the same approach may take many more rounds due to possible congestion in the network (see the next paragraph for the definition of the \congest model). This is one of the reasons why the subgraph detection problem is interesting in the distributed setting. In the last few years, there has been significant progress in understanding the complexity of the non-induced subgraph detection in the \congest model.

\subparagraph*{The \congest model.}
In this paper we use the standard \congest model, as used in prior works~\cite{censor2021tight,censor2019fast,censor2020distributed,chang2019distributed,chang2019improved,czumaj2020detecting,eden2019sublinear,fischer2018possibilities,izumi2017triangle}.
In the \congest model, a distributed network of $n$ computers is represented as a simple undirected graph $G=(V,E)$ of $n$ nodes, where each node corresponds to a computational device, and each edge corresponds to a communication link.
Each node $v\in V$ initially has a $\Theta(\log{n})$-bit unique identifier $\mathrm{ID}(v)$, and knows the list of IDs of its neighbors and the parameter $n=|V|$.
The communication proceeds in synchronous rounds. In each round, each $v\in V$ can perform unlimited local computation, and can send an $O(\log{n})$-bit distinct message to each of its neighbors.

When considering subgraph detection in the \congest model, the communication network is identified with the input graph, i.e., we ask whether the $n$-node communication network contains $H$ as a subgraph. If the network contains $H$ as a subgraph, at least one node outputs~1 (Yes), otherwise all nodes output 0 (No). We assume that each node knows the graph $H$ to be detected. The complexity is characterized by the number of rounds of communication needed to solve the problem.

\subparagraph*{Non-induced subgraph detection in the distributed setting.} 
Typical examples of the non-induced subgraph detection in the \congest model that have been studied intensively  are cliques and cycles. For cliques, the first sublinear-round algorithm of $k$-clique detection in the \congest model is due to Izumi and Le Gall~\cite{izumi2017triangle}, for $k=3$ (i.e., triangle detection), which runs in $\tilde{O}(n^{2/3})$ rounds. 
Later, the complexity was brought down to $\tilde{O}(\sqrt{n})$ by Chang et al.~\cite{chang2019distributed}, and then further to $\tilde{O}(n^{1/3})$ by Chang and Saranurak~\cite{chang2019improved}. These upper bounds also hold for 3-clique listing,\footnote{The listing version of the problem asks to list all instances of $H$ in the graph.} and the $\tilde{O}(n^{1/3})$ upper bound is tight up to polylogarithmic factors due to the lower bounds by Pandurangan, Robinson and Scquizzato~\cite{pandurangan2018distributed} and Izumi and Le Gall~\cite{izumi2017triangle}. For $k$-cliques with $k\geq 4$, the first sublinear algorithm of $k$-clique listing is due to Eden et al.~\cite{eden2019sublinear}. They showed that one can list all $k$-cliques in $\tilde{O}(n^{5/6})$ rounds for $k=4$ and $\tilde{O}(n^{21/22})$ rounds for $k=5$.
These results were improved to $\tilde{O}(n^{k/(k+2)})$ rounds for all $k\geq 4$ by Censor-Hillel, Le Gall, and Leitersdorf~\cite{censor2020distributed}, and very recently, $\tilde{O}(n^{1-2/k})$ rounds for all $k\geq 4$ by Censor-Hillel et al.~\cite{censor2021tight}. The latter bound is tight up to polylogarithmic factors due to the lower bounds by~\cite{fischer2018possibilities}.

For $k$-cycles with $k\geq 4$, it is known that non-induced $k$-cycle ($C_k$) detection requires $\Omega(ex(n,C_k)/n)$ rounds by Drucker et al.~\cite{drucker2014power}, where $ex(n,C_k)$ is the Tur\'{a}n number of $k$-cycle (Tur\'{a}n number $ex(n,C_k)$ is the maximum number of edges in an $n$-node graph which does not have a $k$-cycle as a subgraph). This implies the $\tilde{\Omega}(n)$ lower bounds for odd $k$ and $\tilde{\Omega}(\sqrt{n})$ lower bound for $k=4$. Korhonen and Rybicki~\cite{Korhonen2017deterministic} showed $\tilde{O}(n)$-round algorithms of non-induced $k$-cycle detection for any odd constant $k$. They also showed the $\tilde{\Omega}(\sqrt{n})$ lower bounds for even $k\geq 6$. For $k=4$, an optimal algorithm for non-induced $4$-cycle detection is known due to drucker et al.~\cite{drucker2014power}. For even $k \geq 6$, Fischer et al.~\cite{fischer2018possibilities} showed an $\tilde{O}(n^{1-\frac{1}{k(k-1)}})$-round algorithm, and this was improved to $\tilde{O}(n^{1-2/\Theta(k^2)})$ by Eden et al.~\cite{eden2019sublinear}. Recently, Censor-Hillel et al.~\cite{censor2020fast} showed that for $3\leq k \leq 5$, non-induced $C_{2k}$ detection can be solved in $\tilde{O}(n^{1-1/k})$ rounds. 

We refer to Table~\ref{table:prior work} for the summary of all these results.

\begin{table}[t]\centering
\caption{Prior results for non-induced subgraph detecting and listing in the distributed setting. Here $n$ denotes the number of nodes in the network.}
\scalebox{0.9}[0.9]{ 
  \begin{tabular}{|c|c|c|c|c|} \hline
    \multicolumn{2}{|c|}{Problem} & Time bound & Paper & Model\\ \hline \hline

    \multicolumn{2}{|c|}{\multirow{3}{*}{Triangle detection}}& $\tilde{O}(n^{1/3})$ & \cite{chang2019improved} & {\scriptsize \congest} \\ \cline{3-5}
    
    \multicolumn{2}{|c|}{}& $\tilde{O}(n^{1/4})$ & \cite{izumi2020quantum} & {\scriptsize \qcongest} \\ \cline{3-5}
   
    \multicolumn{2}{|c|}{} & $O(n^{0.159})$ & \cite{censor2019algebraic} & {\scriptsize \clique} \\ \hline
    
    \multicolumn{2}{|c|}{Triangle listing} &$\tilde{\Theta}(n^{1/3})$ & {\scriptsize\cite{izumi2017triangle,chang2019improved,pandurangan2018distributed}}& {\scriptsize \congest} \\ \hline
    
    \multicolumn{2}{|c|}{\multirow{2}{*}{$k$-clique detection, $k\geq 4$}} & $\Omega(n^{1/2}/\log n)$ & \cite{czumaj2020detecting} & {\scriptsize \congest} \\ \cline{3-5}
    
    \multicolumn{2}{|c|}{} & $\tilde{O}(n^{1-2/k})$ & \cite{censor2021tight}& {\scriptsize \congest} \\ \hline
    
    \multicolumn{2}{|c|}{$k$-clique listing, $k\geq 4$} & $\tilde{\Theta}(n^{1-2/k})$ & \cite{fischer2018possibilities,censor2021tight} & {\scriptsize \congest} \\ \hline

    \multirow{5}{*}{$2k$-cycle detection} & $k\geq 2$ & $\Omega(n^{1/2}/\log n)$& \cite{drucker2014power,Korhonen2017deterministic} & {\scriptsize \congest} \\ \cline{2-5}
    
    &$k= 7,9,11,...$ & $\tilde{O}(n^{1-2/(k^2-k + 2)})$ & \cite{eden2019sublinear} & {\scriptsize \congest} \\ \cline{2-5}
  
    &$k= 6,8,10,...$& $\tilde{O}(n^{1-2/(k^2-2k + 4)})$ & \cite{eden2019sublinear} & {\scriptsize \congest} \\ \cline{2-5}
  
    &$2\leq k\leq 5$ & $\tilde{O}(n^{1-1/k})$ & \cite{censor2020fast,drucker2014power} & {\scriptsize \congest} \\ \cline{2-5}
    
    & any constant $k$ & $O(1)$ & \cite{censor2020fast} & {\scriptsize \clique} \\ \hline
    
    \multicolumn{2}{|c|}{$(2k+1)$-cycle detection, $k\geq 2$} & $\tilde{\Theta}(n)$ & \cite{drucker2014power,Korhonen2017deterministic} & {\scriptsize \congest} \\ \hline

    \multicolumn{2}{|c|}{Detecting some $\Theta(k)$-node subgraph $H$}& $\Omega(n^{2-1/k}/\log n)$ & \cite{fischer2018possibilities} & {\scriptsize \congest} \\ \hline
    
  
    \multicolumn{2}{|c|}{Detecting a $k$-node tree} & $O(k^{k})$ & \cite{fraigniaud2017distributed,Korhonen2017deterministic} & {\scriptsize \congest}\\ \hline 
\end{tabular}\label{table:prior work}
}
\end{table}

\subparagraph*{Induced subgraph detection in the distributed setting.}
All of the above results for cycles are only for the non-induced distributed subgraph detection problem. While for the case of cliques, non-induced detection and induced subgraph detection are the same problem, this is not the case for cycles (see again Figure \ref{fig:subgraph_detection} for an illustration). For instance, if we want to know if the input graph contains a chordless cycle, we need to consider the induced version. In the centralized (i.e., non-distributed) setting, the induced version of subgraph detection has thus also been extensively studied \cite{corneil1985linear,eisenbrand2004complexity,itai1978finding,kowaluk2013counting,nevsetvril1985complexity,williams2014finding}, leading to several algorithms that significantly differ  from the algorithms for the non-induced version of the problem.

Despite its importance, the induced version has almost not been studied at all in the distributed setting. The only known results on the complexity of the induced subgraph detection problem in the \congest model are generic bounds describing how large the round complexity can be with respect to the number of nodes in the subgraph: Fischer et al.~\cite{fischer2018possibilities} constructed a family of graphs $H$ with $\Theta(k)$ nodes such that (induced and non-induced) $H$ detection requires $\Omega(n^{2-1/k}/\log{n})$ rounds. Later, Eden et al.~\cite{eden2019sublinear} showed that the $n^{1/k}$ term cannot be removed, and also showed that for any $k$-node subgraph $H$, induced $H$ detection can be solved in $\tilde{O}(n^{2-\frac{2}{3k-2}})=\tilde{O}(n^{2-\Theta(1/k)})$ rounds. These results, however, actually hold for the non-induced version as well. Therefore, to our knowledge, it is still open whether there exists a graph $H$ such that the round complexities of non-induced $H$ detection and induced~$H$ detection are different in the \congest model. 

\subparagraph*{Our results.}
In this paper we answer this question. More precisely, we seek to improve our understanding of the round complexity of the distributed induced subgraph detection in the \congest model by showing lower bounds for constant-length cycles. We refer to Table~\ref{table:our results} for the summary of our results. 

We first show that for any $k\geq 4$, detecting an induced $k$-cycle requires a near-linear amount of rounds; previously, no lower bound for induced cycle detection was known.

\begin{table}[tb]\centering
  \caption{Our results on the round complexity of induced-subgraph detecting, and the corresponding known results. Here $n$ denotes the number of nodes in the network.}
  \begin{tabular}{|c|c|c|c|} \hline
    Problem & Time Bound & Reference & Model \\ \hline \hline
    \!induced $k$-node subgraph detection\!& $\tilde{O}(n^{2-2/(3k+1)})$ & \cite{eden2019sublinear} & \congest\\ \hline
      \multirow{2}{*}{induced $k$-cycle detection}
      & $\Omega(n/\log n)$, for $k\geq 4$ & Theorem~\ref{Theorem:induced_C4} & \congest\\ \cline{2-4}
                                      & $\Omega(n^{2-1/\lfloor k/8 \rfloor}/\log n)$, for $k\geq 8$ & Theorem~\ref{Theorem:induced_Ck} & \congest\\ \hline
    \multirow{2}{*}{induced diamond listing} &$\tilde{O}(n^{5/6})$ & Theorem~\ref{theorem_diamond_upper_bound} &\congest\\ \cline{2-4}
     &$\Omega(\sqrt{n}/\log{n})$ & Theorem 5 &\congest\\ \hline   
  \end{tabular}\label{table:our results}

\end{table}
\begin{theorem}\label{Theorem:induced_C4}
For any $k\geq 4$, deciding if a graph contains an induced $k$-cycle requires $\Omega(n/\log n)$ rounds in the \congest model.
\end{theorem}
\noindent For $k=4$, the trivial solution of induced $k$-cycle detection is to have each node send its entire neighborhood to all its neighbors, which can be done in $O(n)$ rounds. Therefore, our bound in Theorem 1 is tight up to logarithmic factor. Since, as already mentioned, the non-induced version of $4$-cycle detection has complexity $\tilde\Theta(\sqrt{n})$, Theorem 1 proves that the induced version is significantly harder in the \congest model.

We then show stronger lower bounds for induced $C_k$-detection for larger values of $k$.

\begin{theorem}\label{Theorem:induced_Ck}
For any constant $k= 8\ell + m$ where $\ell \geq 1$ and $m\in\{0,1,\ldots,7\}$, deciding if a graph contains an induced $k$-cycle requires $\Omega(n^{2-1/\ell}/\log n)$ rounds in the \congest model, even when the diameter of the network is $3$.
\end{theorem}

\noindent
These bounds are asymptotically tight with respect to $k$, since for any $k$-node subgraph $H$, induced $H$ detection can be solved in $\tilde{O}(n^{2-\Theta(1/k)})$ rounds by the algorithm of~\cite{eden2019sublinear}. We can summarize this as follows.

\begin{corollary}
The round complexity of induced $k$-cycle detection in the \congest model is $\tilde \Theta(n^{2-\Theta(1/k)})$.
\end{corollary}

For small $k$, there still exist gaps between our lower bounds and known upper bounds. For instance, we do not know if induced $5$-cycle detection can be solved in $\tilde{O}(n)$ rounds. This leads to the following question: can we show any improved lower bounds in the case of $k\geq 5$? We complement our results by showing that reductions from two-party communication complexity, which is the technique we used to show our lower bounds (as well as most of the other lower bounds in the literature), 
do not have the ability to derive better lower bounds for the case of $k\leq 7$. 
\begin{theorem}[Informal statement]
For $k=5,6,7$ and any $\varepsilon > 0$, reductions from two-party communication complexity cannot give an $\tilde{\Omega}(n^{1+\varepsilon})$ lower bound for induced $C_k$ detection in the \congest model.
\end{theorem}
\noindent Theorem 4 is shown via an argument similar to the arguments in~\cite{czumaj2020detecting,eden2019sublinear}. These papers showed that reductions from two-party communication complexity (more precisely, the \textit{family of lower bound graphs} technique) cannot show $\tilde{\Omega}(n^{1/2+\varepsilon})$ lower bounds of $4$-clique detection and (non-induced) $6$-cycle detection. As mentioned above, to date, all known lower bounds on the subgraph detection in the \congest model used reductions from two-party communication complexity. Therefore, Theorem 4 shows that we need a fundamentally different approach to improve these lower bounds.



The graph that is obtained by removing one edge from a 4-clique is called a \textit{diamond}. Diamonds are interesting since it is in some sense intermediate between 4-cycle and 4-clique. We show a lower bound of induced diamond listing in the CONGEST model.
\begin{theorem}\label{theorem_diamond}
Listing all induced diamonds requires $\Omega(\sqrt{n}/\log n)$ rounds in the \congest model.
\end{theorem}
We also prove the same result as in Theorem 4 for induced diamond listing.
\begin{theorem}[Informal statement]\label{theorem_diamond_limitation}
For any $\varepsilon > 0$, reductions from two-party communication complexity cannot give an $\tilde{\Omega}(n^{1/2+\varepsilon})$ lower bound for induced diamond listing in the \congest model.
\end{theorem}
Finally, we show that induced diamond listing can be done in sublinear rounds. 
\begin{theorem}\label{theorem_diamond_upper_bound}
There exists an algorithm that solves induced diamond listing in $\tilde{O}(n^{5/6})$ rounds in the \congest model.
\end{theorem}\vspace{2mm}
Due to space constraints, the proofs of Theorem \ref{theorem_diamond_limitation} and Theorem \ref{theorem_diamond_upper_bound} are omitted from the main body of this paper --- they can be found in Appendix \ref{appC} and Appendix~\ref{appD}.

\subparagraph*{Other related works.} Several works investigate the complexity of subgraph detection in other models of distributed computing.
In the powerful \clique model, which allows global communication, the induced subgraph detection (and even listing) can be solved in sublinear rounds~\cite{dolev2012tri}: for any $k$-node subgraph $H$, induced $H$ detection and listing can be solved in $O(n^{1-2/k})$ rounds.
This algorithm was used as a subroutine to construct sublinear-round \congest algorithms for clique detection and listing~\cite{censor2021tight,censor2020distributed,chang2019distributed,chang2019improved,izumi2020quantum,izumi2017triangle}.
In the \clique model, for any constant $k\geq 3$, $k$-cycle can be detected in $O(2^{O(k)}n^{0.158})$ rounds by an algebraic algorithm which uses the matrix multiplication~\cite{censor2019algebraic} and for $k\geq 2$, $2k$-cycle can be detected in O(1) rounds~\cite{censor2020fast}. In the \qcongest model, in which each node represents a quantum computer and each edge represents a quantum channel, Izumi, Le Gall and Magniez~\cite{izumi2020quantum} showed that triangle detection can be solved in $\tilde{O}(n^{1/4})$ rounds by using quantum distributed search~\cite{le2018sublinear} which is the distributed implementation of Grover's quantum search. For any $\varepsilon > 0$, showing a lower bound of $\Omega(n^{\varepsilon})$ on directed triangle detection implies strong circuit complexity
lower bounds~\cite{censor2020fast}. These bounds are included in Table~\ref{table:prior work}.

Other relevant works include constant-round detection of  constant-sized trees in the \congest model~\cite{fraigniaud2017distributed,Korhonen2017deterministic},  
and investigations of the distributed subgraph detection problem in the framework of \textit{property testing}~\cite{censor2019fast,fraigniaud2019distributed,fraigniaud2016distributed}.

\subparagraph*{Recent independent work.}
Lower bounds for induced cycle detection similar to some of the bounds in our paper have been concurrently (and independently) obtained very recently by Korhonen and Nikabadi \cite{korhonen2021distributed}.
They showed the following results using graph constructions different from ours (but still using reductions from two-party communication complexity):
\begin{itemize}
    \item $\Omega(n/\log n)$ lower bound for induced $2k$-cycle detection for $k\geq 3$,
    \item $\Omega(n/\log n)$ lower bound for a multicolored variant of $k$-cycle detection for $k\geq 4$.
\end{itemize}

\section{Preliminaries}

To prove lower bounds, we use reductions from two-party communication complexity problems. This is the common technique to show lower bounds in the \congest model~\cite{fischer2018possibilities,censor2017quadratic,czumaj2020detecting,drucker2014power,abboud2016near}. Here we give the precise definition of the \textit{family of lower bound graphs}, which is the standard notion to show these lower bounds (see, e.g.,~\cite{censor2017quadratic}).
\begin{definition}[Family of Lower Bound Graphs]\label{def:LBGraph}
Given an integer $K$, a boolean function $f:\{0,1\}^K \times \{0,1\}^K\rightarrow \{0,1\}$ and a graph predicate $P$, a set of graphs $\{G_{x,y}=(V,E_{x,y})|x,y\in\{0,1\}^K\}$ is called \textit{a family of lower bound graphs} with respect to $f$ and $P$
 if the following hold:
 \begin{enumerate}
     \item The set of vertices $V$ is the same for all the graphs in the family, and has a fixed partition $V=V_A\cup V_B$. The set of edges of the cut $E_{cut}=E(V_A,V_B)$ is the same for all graphs in the family.
     \item Given $x,y\in\{0,1\}^K$, $E(V_A,V_A)$ only depends on $x$.
     \item Given $x,y\in\{0,1\}^K$, $E(V_B,V_B)$ only depends on $y$.
     \item $G_{x,y}$ satisfies $P$ if and only if $f(x,y) = 1$.
 \end{enumerate}
 \end{definition}
 \begin{theorem}[\cite{censor2017quadratic}]\label{theo:Family_of_Lower_Bound_Graphs}
 Fix a boolean function $f:\{0,1\}^K \times \{0,1\}^K\rightarrow \{0,1\}$ and a graph predicate $P$. If there exists a family of lower bound graphs $\{G_{x,y}\}$ with respect to $f$ and $P$, then any randomized algorithm for deciding $P$ in the \congest model takes $\Omega(CC^{R}(f)/|E_{cut}|\log n)$, where $CC^R(f)$ is the randomized communication complexity of $f$.
 \end{theorem}
 
\noindent Throughout this paper, we use the set-disjointness function $\mathrm{DISJ}_{K}:\{0,1\}^K \times \{0,1\}^K\rightarrow \{0,1\}$. For two bit strings $x,y\in\{0,1\}^K$, $\mathrm{DISJ}_{K}(x,y)$ is equal to $0$ if and only if there exists some index $i\in [K]$ such that $x_i=y_i=1$. It is well known that $CC^R(\mathrm{DISJ}_K)=\Omega(K)$~\cite{razborov1990distributional}.

\section{Lower Bounds for $k$-cycles, $k\geq 4$.}
In this section we prove Theorem~\ref{Theorem:induced_C4}. To prove Theorem~\ref{Theorem:induced_C4}, we describe families of lower bound graphs with respect to the set-disjointness function of the two-party communication complexity, and the predicate $P$ that says the graph does not contain an induced $k$-cycle. We start by describing the fixed graph construction for the case of $k=4$, and then define the corresponding family of lower bound graphs.
\subparagraph*{The fixed graph construction.}
Create a graph $G$ as follows:
The vertex set is $A_1\cup A_2 \cup B_1\cup B_2$ such that $A_1$ and $B_2$ are $n$-vertex cliques and $A_2$ and $B_1$ are a set of $n$ vertices with no edges inside of them.
Denote the vertices in $G$ as $A_i=\{a_i^1,\ldots,a_i^n\}$, $B_i=\{b_i^1,\ldots,b_i^n\}$ for $i\in\{1,2\}$. We add edges $(a_1^i,b_1^i), (a_2^i,b_2^i)$ for all $i\in[n]$.
\subparagraph*{Creating $G_{x,y}$.}
\begin{figure}[tbp]
\centering
    \includegraphics[width=0.65\hsize]{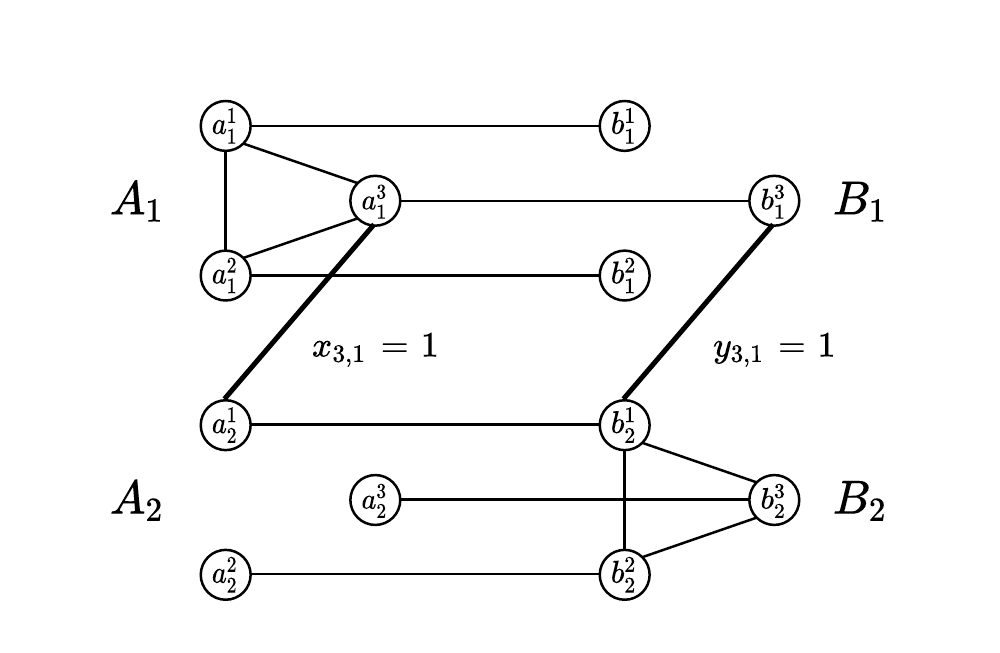}\vspace{-2mm}
    \caption{An illustration of $G_{x,y}$ for the case of $n=3$.
    The graph contains a copy of induced $C_4$ if and only if it holds that $x_{ij}=y_{ij}=1$ for some index $i,j\in [n]$. This illustration shows the case of $x_{3,1}=y_{3,1}=1$.} 
    \label{fig:LBforCk}\vspace{-3mm}
\end{figure}
For two input bit strings $x, y\in\{0,1\}^{n^2}$ and $i,j\in [n]$, we denote the ($i+(j-1)n$)-th bit of $x$ and $y$ as $x_{ij}$ and $y_{ij}$.
Edges corresponding to inputs are added as follows (see Figure~\ref{fig:LBforCk} for the illustration):
\begin{itemize}
    \item We add an edge between $a_1^i$ and $a_2^j$ if and only if $x_{ij}=1$.
    \item We add an edge between $b_1^i$ and $b_2^j$ if and only if $y_{ij}=1$.
\end{itemize}
This concludes the description of $G_{x,y}$. Next, we prove that the family $\left\{G_{x,y}|x,y\in \{0,1\}^{n^2}\right\}$ is a family of lower bound graphs with respect to set-disjointness and the predicate that says the graph does not contain an induced $4$-cycle.

\begin{claim}\label{claim_C4}
$G_{x,y}$ contains an induced $C_4$ if and only if there exists a pair of index $i, j\in[n]$ such that $x_{ij}=y_{ij}=1$.
\end{claim}
\begin{proof}
Let $U=\{v_1,v_2,v_3,v_4\}$ be a subset of $V$.
It is clear that if it holds $|U\cap S|=4$ for some $S\in\{A_1,A_2,B_1,B_2\}$, then $U$ does not induce $C_4$. We analyse $U$ as follows:
\begin{itemize}
    \item If it holds $|U\cap A_1|=3$ or $|U\cap B_2|=3$, $U$ induces a triangle.
    \item If it holds $|U\cap A_2|=3$ or $|U\cap B_1|=3$, $U$ induces at most three edges.
    \item If it holds $|U\cap A_1|=2$ or $|U\cap B_2|=2$, and $U$ induces a $4$-cycle, the other two vertices of $U$ are both in $A_2$ or both in $B_1$. However, It is impossible since there is no edge between any two vertices in $A_2$ and any two vertices in $B_1$. 
    \item If it holds $|U\cap A_2|=2$ or $|U\cap B_1|=2$, the two vertices does not share neighbors. Hence, $U$ does not induce a $4$-cycle.  
\end{itemize}
Now all we need is to verify whether four vertices $a_1^i$, $a_2^j$, $b_1^k$, and $b_2^{\ell}$
induce a $4$-cycle or not.
If $(a_1^i,a_2^j,b_1^k,b_2^{\ell})$
induces a $4$-cycle, we can say that $i=k$ since $a_1^i$ has to be connected to $a_2^j$ and $b_1^k$ (similarly we can say $j=\ell$). It is straightforward to show that $(a_1^i,a_2^j,b_1^i,b_2^{j})$
induces a $4$-cycle if and only if $x_{ij}=y_{ij}=1$.
\end{proof}
\subparagraph*{Proof of Theorem \ref{Theorem:induced_C4}:}
Divide the vertices of the graph $G_{x,y}$ into $V_A=A_1\cup A_2$ and $V_B=B_1\cup B_2$. The size of the cut is $|E_{cut}|=|E(V_A,V_B)|=2n$.
Claim \ref{claim_C4} shows that the family of the graphs $\left\{G_{x,y}\middle| x,y\in\{0,1\}^{n^2}\right\}$ is a family of lower bound graphs for $f=\mathrm{DISJ}_{n^2}$ and a predicate that says the graph include an induced $C_4$.
Hence, using Theorem \ref{theo:Family_of_Lower_Bound_Graphs} and $CC^R(\mathrm{DISJ}_{n^2}) = \Theta(n^2)$, any randomized algorithms for induced $4$-cycle detection in the \congest model requires $\Omega(n/\log n)$.

To extend this result to $k$-cycles for $k\geq 5$, we modify the graph $G_{x,y}$ as follows:
\begin{itemize}
    \item For any $i\in [n]$, replace the edge $(a_1^i,b_1^i)$ to a path with $\lceil\frac{k-4}{2}\rceil + 2$ vertices. 
    \item For any $i\in [n]$, replace the edge $(a_2^i,b_2^i)$ to a path with $\lfloor\frac{k-4}{2}\rfloor + 2$ vertices. 
\end{itemize}
\qed



\section{Lower Bounds for Larger Cycles}\label{sec:C8}

In this section we prove Theorem~\ref{Theorem:induced_Ck}, i.e., 
we show subquadratic, but superlinear lower bounds for induced cycles $C_{k\geq 8}$, which gives nearly tight bounds for induced cycles $C_{k\geq 8}$ with respect to $k$.
The main difficulty to obtain the bounds of Theorem~\ref{Theorem:induced_Ck} is to reduce the size of the cut edges of graphs while retaining the ability to simulate the set-disjointness function of size $\Omega(n^2)$. We overcome this difficulty by considering induced cycles that go around $G_{x,y}$ more than once instead of cycles that go around $G_{x,y}$ exactly once (we pay for this by an increased size of a cycle). This enables us to reduce the size of cut edges.

\subsection{The fixed graph construction}\par
\begin{figure}[tbp]\centering
    \includegraphics[width=0.7\hsize]{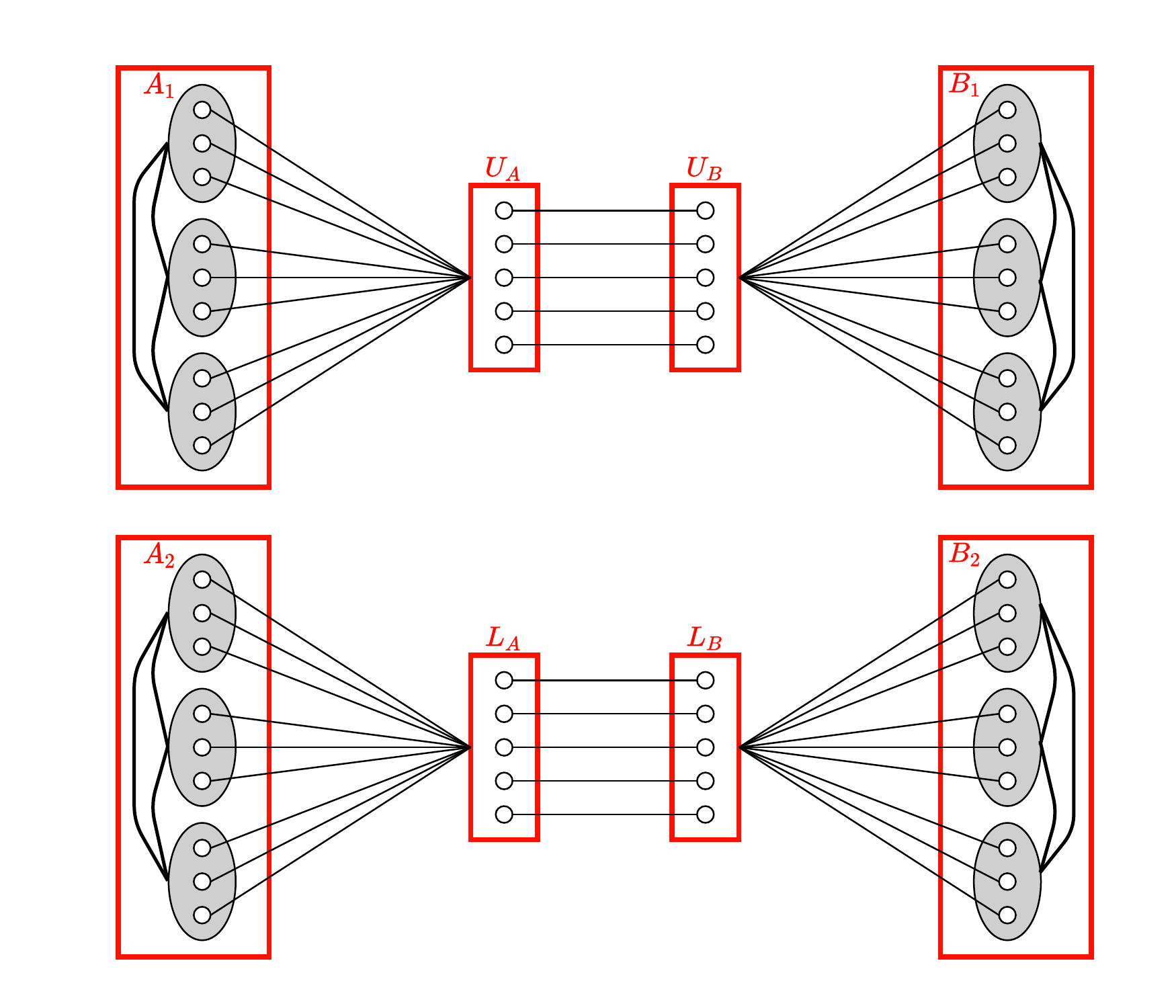}
    \caption{An illustration of the fixed part of $G_{x,y}$. Some edges are bundled for clarity. Observe that $A_1^i\subseteq A_1$ and $A_2^j\subseteq A_2$ are connected by additional edges iff $x_{i,j}=1$. Also, $B_1^i\subseteq B_1$ and $B_2^j\subseteq B_2$ are connected by additional edges iff $y_{i,j}=1$.\vspace{1cm}} 
    \label{fig:LB graph for C8}\vspace{-10mm}
\end{figure}

We refer to Figure \ref{fig:LB graph for C8} for an illustration of the construction.
\subparagraph*{Vertices.}
We define the sets of vertices as follows:

\begin{itemize}
    \item $A_k=A_k^1\cup\cdots \cup A_k^{n}$, where
    $A_k^i=\left\{ a_k^{i,j} \middle| 0\leq j \leq \ell-1 \right\}$ for $i\in [n]$ and $k\in\{1,2\}$.
    \item $B_k=B_k^1\cup\cdots\cup B_k^{n}$, where
    $B_k^i=\left\{ b_k^{i,j} \middle| 0\leq j \leq \ell-1 \right\}$ for $i\in [n]$ and $k\in\{1,2\}$.
    \item
    $U_A=\left\{ u_A^{i} \middle| 0\leq i \leq \ell n^{1/\ell} \right\}$, $L_A=\left\{ l_A^{i} \middle| 0\leq i \leq \ell n^{1/\ell} \right\}$.
    \item 
    $U_B=\left\{ u_B^{i} \middle| 0\leq i \leq \ell n^{1/\ell} \right\}$, $L_B=\left\{ l_B^{i} \middle| 0\leq i \leq \ell n^{1/\ell} \right\}$.
\end{itemize}

\noindent Each $S\in \{A_1,A_2,B_1,B_2\}$ contains $\ell n$ vertices, and they are divided into $n$ subsets of size $\ell$.
Each $C\in \{U_A,U_B,L_A,L_B\}$ contains $\ell n^{1/\ell}$ vertices. The number of vertices $V=A_1\cup A_2\cup B_1\cup B_2 \cup U_A \cup L_A \cup U_B \cup L_B$ is $\Theta(\ell n + \ell n^{1/\ell}) = \Theta(n)$.
\subparagraph{Edges.}
First, we add $2\ell n^{1/\ell}$ edges $\left\{(u_A^i,u_B^i),(l_A^i,l_B^i)\middle| i\in [\ell n^{1/\ell}]\right\}$.
Then, we consider a map from $[n]$ to $[\ell n^{1/\ell}]^\ell$, where $[\ell n^{1/\ell}]^\ell$ is $\ell$ times direct product of the set $[\ell n^{1/\ell}]$.
Since
\begin{align*}
    \begin{pmatrix}
    \ell n^{1/\ell}\\
    \ell
    \end{pmatrix}
    = \frac{\ell n^{1/\ell}}{\ell} \cdot \frac{\ell n^{1/\ell}-1}{\ell-1} \cdots \frac{\ell n^{1/\ell}-\ell + 1}{1}
    \geq \left(\frac{\ell n^{1/\ell}}{\ell}\right)^{\ell} = n
\end{align*}
holds, there exists an injection $\sigma:[n]\rightarrow [\ell n^{1/\ell}]^{\ell}$.
We arbitrarily choose one of these injections.
For $i\in [n]$, we denote $\sigma(i)=\{k_1,\ldots ,k_{\ell}\}\in [\ell n^{1/\ell}]^{\ell}$.
For all $i \in [n],j\in [\ell]$, we add the edge sets 
$\left\{ (a_1^{i,j},u_A^{k_j}) \middle| i\in [n],j\in [\ell] \right\}$,
$\left\{ (a_2^{i,j},l_A^{k_j}) \middle| i\in [n],j\in [\ell] \right\}$,
$\left\{ (b_1^{i,j},u_B^{k_j}) \middle| i\in [n],j\in [\ell] \right\}$, and
$\left\{ (b_2^{i,j},l_B^{k_j}) \middle| i\in [n],j\in [\ell] \right\}$.
Now we can determine exactly $\ell$ vertices of $U_A$ that are adjacent to vertices of $A_1^i$. We denote them $Code(A_1^i)\subseteq U_A$. In the same way, we determine the vertex sets $Code(A_2^i)\subseteq L_A$, $Code(B_1^i)\subseteq U_B$, and $Code(B_2^i)\subseteq L_B$ by using the same $\sigma$.
Since $\sigma$ is an injection, it holds that $Code(A_1^i)\neq Code(A_1^j)$, $Code(A_2^i)\neq Code(A_2^j)$, $Code(B_1^i)\neq Code(B_1^j)$, and $Code(B_2^i)\neq Code(B_2^j)$ for $i\neq j$. 

\noindent In addition, we add the following edges.
\begin{itemize}
    \item For any $i,j\in [n]$, add edges between $u\in A_1^i,v\in A_1^j$ if and only if $i\neq j$.
    \item For any $i,j\in [n]$, add edges between $u\in B_2^i,v\in B_2^j$ if and only if $i\neq j$.
\end{itemize}
If $\ell \geq 2$, we add the following edges.
\begin{itemize}
    \item For any $i,j\in [n]$, add edges between $u\in A_2^i,v\in A_2^j$ if and only if $i\neq j$.
    \item For any $i,j\in [n]$, add edges between $u\in B_1^i,v\in B_1^j$ if and only if $i\neq j$.
\end{itemize}

\subsection{Creating $G_{x,y}$}
Note that for $\ell = 1$, the fixed part of $G_{x,y}$ in this section is exactly the same as the fixed part of graphs for induced $8$-cycles in Section 3. Hence, we only describe the case $\ell \geq 2$.
Given two binary strings $x,y\in\{0,1\}^{n^2}$, we add the following edges:
\begin{itemize}
    \item For $i,j\in [n]$, add edges $\{(a_1^{i,k+1},a_2^{j,k})|k\in [\ell-1]\} \cup \{(a_1^{i,1},a_2^{j,\ell})\}$, if and only if $x_{i,j}=1$.
    \item For $i,j\in [n]$, add edges $\{(b_1^{i,k},b_2^{j,k})|k\in [\ell]\}$, if and only if $y_{i,j}=1$.
\end{itemize}

This concludes the description of $G_{x,y}$. We show the following theorem which says that $\left\{G_{x,y}\right\}$ is a family of lower bound graphs. Due to space constraint, the proof is moved to the Appendix.

\begin{theorem}\label{LB_for_C8}
$G_{x,y}$ contains an induced $8\ell$-cycle if and only if $\mathrm{DISJ}_{n^2}(x,y)=0$.
\end{theorem}

Having constructed a family of lower bound graphs, we are now ready to prove Theorem \ref{Theorem:induced_Ck}.
\subparagraph*{Proof of Theorem \ref{Theorem:induced_Ck}:}
Theorem \ref{LB_for_C8} implies that a family of graphs 
\begin{align*}
    \left\{G_{x,y}=(V_A\cup V_B,E_{x,y})\middle|x,y\in \{0,1\}^{n^2}\right\}
\end{align*} where $V_A=A_1\cup A_2 \cup U_A \cup L_A$, $V_B=B_1\cup B_2 \cup U_B \cup L_B$ is a family of lower bound graphs with respect to the set disjointness function $\mathrm{DISJ}_{n^2}$ and the graph predicate is whether the graph has a copy of an induced $C_{8\ell}$ or not with cut size $\ell n^{1/\ell}$.
To bound the diameter of the network to $3$, we add nodes $c_A$ to $V_A$ and $c_B$ to $V_B$ such that $c_A$ is connected to all nodes in $V_A$ and $c_B$ is connected to all nodes in $V_B$. Finally, we add an edge $(c_A,c_B)$.
The above modification does not effect to the existence of induced $8\ell$-cycles: If we choose $c_A$ as one of the cycle nodes, then we cannot choose more than two $V_A$ nodes as cycle nodes. However, we cannot choose more than $8\ell - 4$ nodes from $V_B$ due to Lemma \ref{C8l_lemma} of Appendix \ref{appA}, which also holds after this modification.
The theorem is proved by applying Theorem \ref{theo:Family_of_Lower_Bound_Graphs} (for $m=0$). Slightly modifying the graphs gives the same complexity for the case of $k=8\ell+m$ where $m\in\{1,2,\ldots ,7\}$: 
\begin{itemize}
    \item Replace each edge $e\in U_A\times U_B$ by a path of length $\lfloor m/2\rfloor$.
    \item Replace each edge $e\in L_A\times L_B$ by a path of length $\lceil m/2\rceil$.
\end{itemize}
\qed

\section{Limitation of the Two-Party Communication Framework}
Since no $\tilde{O}(n)$-round algorithm for detecting an induced $k$-cycle for $k\geq 5$ is known, the main question is whether our lower bound can be improved or not.
In this section, we show that the family of lower bound graphs cannot derive any better lower bounds for detecting an induced $k$-cycle for $k\leq 7$, by giving a two-party communication protocol for listing $k$-cycles for $k\leq 7$ in the vertex partition model which is defined as follows.
\begin{definition}[Vertex Partition Model,~\cite{czumaj2020detecting}]
Given a graph $G=(V_A\cup V_B,E_A\cup E_B\cup E_{cut})$ where $E_A=E(V_A,V_A), E_B=E(V_B,V_B)$ and $E_{cut}=E(V_A,V_B)$, \textit{the vertex partition model} is a two-party communication model in which Alice receives $G_A=(V_A,E_A\cup E_{cut})$ as the input and Bob receives $G_B=(V_B,E_B\cup E_{cut})$ as the input.  
For any graph $H$, the $O(k)$ communication protocol for induced $H$ listing is a protocol such that
\begin{itemize}
    \item The players communicate $O(k)$ bits in the protocol.
    \item At the end of the protocol, the players have their lists of $H$, denoted by $A_H,B_H$, such that all of copies of $H$ in the input graph $G$ are contained in either $A_H$ or $B_H$.
\end{itemize}
\end{definition}

\begin{theorem}\label{theorem:cycle_protocol}
There is a two-party communication protocol in the vertex partition model for listing all induced $k$-cycles for $k\leq 7$ that uses $\tilde{O}(n|E_{cut}|)$ bits of communication where $E_{cut}$ is the set of cut edges in the input graph.
\end{theorem}

\begin{proof}
Let $V'_A(V'_B)$ be a set of $V_A(V_B)$ vertices which are incident to some cut edge. 
The protocol is as follows:
\begin{enumerate}
    \item Bob sends all edges $E_B\cap \{V'_B\times V_B\}$ to Alice in $\tilde{O}(n|E_{cut}|)$ bits since the number of edges Bob sends to Alice is less than
\begin{align*}
    \sum_{v\in V'_B} \text{deg}_{V_B}(v)\leq \sum_{v\in V'_B} n =n|E_{cut}|.
\end{align*}
\item Alice sends all edges $E_A\cap \{V'_A\times V_A\}$ to Bob in $\tilde{O}(n|E_{cut}|)$ bits since the number of edges Alice sends to Bob is less than
\begin{align*}
    \sum_{v\in V'_A} \text{deg}_{V_A}(v)\leq \sum_{v\in V'_A} n =n|E_{cut}|.
\end{align*}
\end{enumerate}
Consider Alice has to list all induced $k$-cycles such that at least $\lceil k/2 \rceil$ vertices of them are in $V_A$. 
Let $U$ be a set of vertices in a copy of an induced $k$-cycle Alice should list in the input graph $G$.
Since $k\leq 7$, $U$ contains at most three vertices in $V_B$.
If $U$ has at least one vertex in $V_B$, then the $k$-cycle induced by $U$ has two cut edges. Therefore, we have $U\times U\subseteq E_A\cup E_{cut}\cup \left\{E_B\cap\{V'_B\times V_B\}\right\}$. Step 1 of the protocol enables Alice to list all induced $k$-cycles she should list. Similarly, step 2 of the protocol enables Bob to list all induced $k$-cycles he should list. Now all induced $k$-cycles of the input graph $G$  are in either the list of Alice or the list of Bob.
\end{proof}

\setcounter{theorem}{3}
\begin{theorem}[Formal statement]\label{theorem_ck_limitation}
For any $\varepsilon > 0$, no family of lower bound graphs gives an $\tilde{\Omega}(n^{1+\varepsilon})$ lower bound of induced $k$-cycle detection for $k\leq 7$.
\end{theorem}\setcounter{theorem}{13}
\begin{proof}
For the family of lower bound graphs \begin{align*}
    \left\{G_{x,y}=(V_A\cup V_B,E_A\cup E_B\cup E_{cut})\middle|x,y\in\{0,1\}^K\right\}
\end{align*}for $f:\{0,1\}^K\times \{0,1\}^K \rightarrow \{0,1\}$ and the property which says that the graph contains an induced $k$-cycle, we can show an $\tilde{\Omega}(CC^R(f)/|E_{cut}|)$ lower bound for induced $k$-cycle detection. On the other hand, we can solve $f$ by $\tilde{O}(n|E_{cut}|)$ bits of communication through the protocol of the vertex partition model in Theorem~\ref{theorem:cycle_protocol}. Then it holds $|E_{cut}|=\tilde{\Omega}( CC^R(f)/n)$, implying that for any $\varepsilon > 0$, we cannot derive an $\tilde{\Omega}(n^{1+\varepsilon})$ lower bound for induced $k$-cycle listing by the family of lower bound graphs.
\end{proof}

\section{Lower Bound for Diamond Listing}
We know that the round complexity of $4$-clique detection is $\tilde{\Theta}(\sqrt{n})$, and  induced $4$-cycle detection is $\tilde{\Theta}(n)$. 
The only four-node graph that lies between $4$-clique and $4$-cycle is the \textit{diamond}, which is the four-node graph obtained by removing one edge from a $4$-clique.
Intuitively, the complexity of diamond detection seems to be somewhere between the complexity of $4$-clique and the complexity of $4$-cycle.
In this section, we make this intuition precise, and we show a lower bound for induced diamond listing (this result is complemented by the upper bound of Theorem \ref{theorem_diamond_upper_bound} shown in Appendix \ref{appD}). Our construction of the family of lower bound graphs is similar to \cite{czumaj2020detecting}, but has the following differences.
The graphs of \cite{czumaj2020detecting} have two sets of vertices $A$ and $B$, and edges between $A$ and $B$ are added randomly so that the size of the cut edges is $O(n^{3/2})$.
In this random graph, \textit{w.h.p.}, the number of tuples of the form $(a_1,a_2,b_1,b_2)$ where $a_1,a_2\in A$ and $b_1,b_2\in B$ which corresponds to the $i$-th bit of the input strings $x,y$ is $\Omega(n^2)$: a tuple $(a_1,a_2,b_1,b_2)$ induces a $4$-clique if and only if $x_i=y_i=1$.
However, in the case of diamonds, the number of tuples which correspond to the input is $o(n^2)$. To avoid this, we construct the family of lower bound graphs in a different way. This makes it much easier to analyze the properties of the graph.
\subparagraph*{The fixed graph construction.}
We refer to Figure \ref{fig:LB} for an illustration.
\begin{figure}[tbp]\centering
    \includegraphics[width=0.65\hsize]{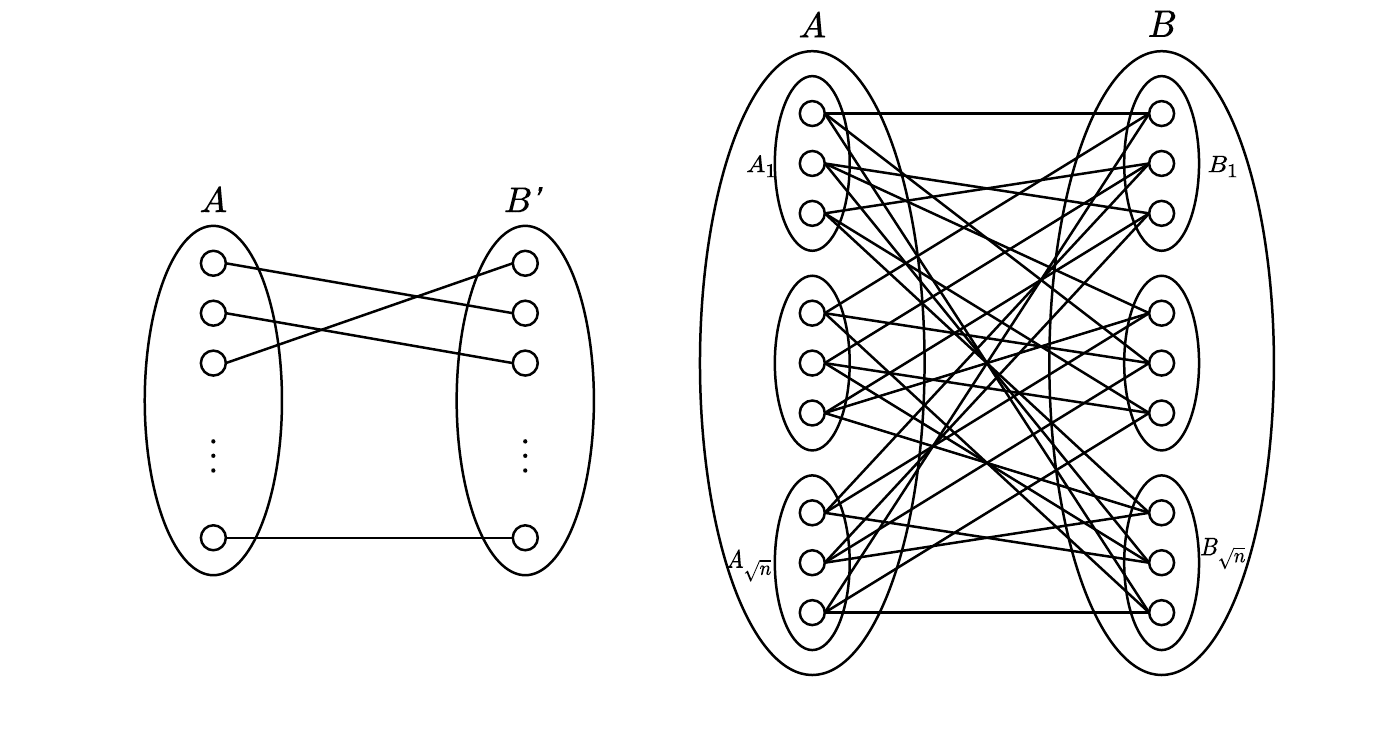}\vspace{-5mm}
    \caption{The cut edges in the family of lower bound graphs for listing diamonds. Many edges are omitted for clarity.} 
    \label{fig:LB}\vspace{-5mm}
\end{figure}
The set of vertices is $V=A\cup B\cup B'$ such that $A,B$ and $B'$ are sets of $n$ vertices. Each vertex is denoted as follows:
\begin{itemize}
    \item $A=A_1\cup A_2\cup \cdots \cup A_{\sqrt{n}}$, where $A_i = \left\{a_i^j\middle| j\in[\sqrt{n}]\right\}$ for all $i\in[\sqrt{n}]$.
    \item $B=B_1\cup B_2\cup \cdots \cup B_{\sqrt{n}}$, where $B_i = \left\{b_i^j\middle| j\in[\sqrt{n}]\right\}$ for all $i\in[\sqrt{n}]$.
    \item $B'=\left\{b'_{i}\middle| i\in [n] \right\}$.
\end{itemize}
We add the edge set $\left\{(a_i^j, b'_{j+(i-1)\sqrt{n}}) \middle|i,j\in [\sqrt{n}] \right\}$.
For any $i,j\in [\sqrt{n}]$, we choose a bijection uniform randomly from all possible bijection $\sigma:A_i\rightarrow B_j$, and denote it $\sigma_{i,j}:A_i\rightarrow B_j$.
Then, we add the edge set 
$
    \left\{ (a_i^k,\sigma_{i,j}(a_i^k))\middle| i,j\in [\sqrt{n}], k\in[\sqrt{n}]  \right\}.
$
\subparagraph*{Creating $G_{x,y}$.}
We call a pair $(a_i^j, a_k^\ell)$ is good iff $|N(a_i^j)\cap N(a_k^\ell)|=1$, where $N(u)$ is a set of neighbors of a vertex $u$.
Let $P_A$ be the set of good pairs in $A\times A$. That is, $P_A:=\left\{ (a_i^j,a_k^\ell)\middle||N(a_i^j)\cap N(a_k^\ell)|=1, i,j,k,\ell\in[\sqrt{n}] \right\}$.

\begin{lemma}\label{diamond_lemma}
There is a graph $G$ created by the above procedure, in which it holds that $|P_A|=\Omega(n^2)$. 
\end{lemma}
The proof of Lemma \ref{diamond_lemma}
can be found in Appendix \ref{appB}.
Consider the graph $G$ in which $|P_A|=\Omega(n^2)$.
Let $\mathcal{H}=\emptyset$.
We partition $A$ randomly into two sets $A^*$ and $A\backslash A^*$ so that $|A^*|=n/2$. 
For a pair $(a_1,a_2)\in P_A$, there is only one vertex $b_1\in N(a_1)\cap N(a_2)$. 
For $a_1$, there is only one vertex $b_2 \in N(a_1)\cap B'$.
We add a quadruple $(a_1,a_2,b_1,b_2)$ to $\mathcal{H}$ iff $|\{a_1,a_2\}\cap A^*| = 1$. This condition holds with probability $\frac{1}{2}$. Then, we remove $(a_1,a_2)$ from $P_A$. We continue this operation until $P_A$ becomes the empty set.
Therefore, after this process, there are $|\mathcal{H}|=\Omega(n^2)$ quadruples in $\mathcal{H}$ with high probability, using Chernoff bound.
We label $\mathcal{H}$ as $\mathcal{H}=\{ h_1,h_2,...,h_{|\mathcal{H}|} \}$ and relabel quadruples as $h_k=(a_{k,1},a_{k,2},b_{k,1},b_{k,2})$.
Consider two bit strings $x,y\in\{0,1\}^{|\mathcal{H}|}$. 
We create a graph $G_{x,y}$ by adding edges to $G$ as follows:
\begin{itemize}
    \item If $x_k=1$, we add an edge between $a_{k,1}$ and $a_{k,2}$.
    \item If $y_k=1$, we add an edge between $b_{k,1}$ and $b_{k,2}$.
\end{itemize}For a quadruple $D=(u_1,u_2,u_3,u_4)$ of vertices, we say that $D$ is an $(i,j)$-diamond when
\begin{itemize}
    \item $(u_1,u_2,u_3,u_4)$ induces a diamond,
    \item $|A\cap \{u_1,u_2,u_3,u_4\}|=i$ and $|(B\cup B')\cap \{u_1,u_2,u_3,u_4\}|=j$.
\end{itemize}
\begin{lemma}\label{lemma_diamond}
$G_{x,y}$ contains a (2,2)-diamond if and only if there exists a pair of indices $i,j\in[\sqrt{|\mathcal{H}|}]$ such that $x_{ij}=y_{ij}=1$.
\end{lemma}
\begin{proof}
It is clear that if $x_k=y_k=1$, then $h_k=(a_{k,1},a_{k,2},b_{k,1},b_{k,2})$ induces a diamond.
Consider four vertices $a_1,a_2 \in A, b_1,b_2\in B\cup B'$ that induce a (2,2)-diamond.
If it holds that $(b_1,b_2)\notin E$, then $(a_1,a_2)\in E$. Thus, a pair $(a_1,a_2)$ is good. It contradicts that vertices $a_1$, $a_2$,  $b_1$, and $b_2$ induce a diamond.
Assume that $(b_1,b_2)\in E$. Without loss of generality, we assume $b_1\in B$, $b_2\in B'$. Since a pair $(a_1,a_2)$ is good, $b_1= N(a_1)\cap N(a_2)$ and $(a_1, a_2)\in E$.
Then, there is an index $k\in [|\mathcal{H}|]$ such that 
$h_k=(a_1,a_2,b_1,b_2)$ since $(a_2,b_2)\notin E$ holds.
\end{proof}
\subparagraph*{Proof of Theorem \ref{theorem_diamond}:}
Consider that Alice and Bob construct the graph $G_{x,y}$ where $V_A=A, V_B=B\cup B'$. By simulating an $r$-round \congest algorithm $\mathcal{A}$ that solves listing all diamonds, they can compute the set disjointness function of size $|\mathcal{H}|=\Omega(n^2)$: Bob tells Alice if there exists a (2,2)-diamond in the output of vertices simulated by Bob by sending 1 bit.  Then, from Lemma \ref{lemma_diamond}, Alice knows $\mathrm{DISJ}_{|\mathcal{H}|}(x,y)$ since Alice can know  whether $G_{x,y}$ contains a (2,2)-diamond. 
The number of edges between Alice and Bob is $n^{3/2} + n = \Theta(n^{3/2})$. Hence, $O(rn^{3/2} \log n) =\Omega(|\mathcal{H}|)$ and this means $r=\Omega(\sqrt{n}/\log n)$.  \qed



\bibliography{index}

\appendix
\section{Proof of Theorem~\ref{LB_for_C8}}\label{appA}
We first show the following two lemmas.

\begin{lemma}\label{C8l_lemma}
Any subset of vertices $\mathcal{C}\subseteq V$ of size $8\ell$ in $G_{x,y}$ which induces $C_{8\ell}$ contains at most $\ell$ vertices in $S$ where $S\in \{A_1,A_2,B_1,B_2\}$.
\end{lemma}
\begin{proof}
It is enough to check the case of $S=A_1$. Let $c$ be the number of index $i\in [n]$ such that $|\mathcal{C}\cap A_1^i| > 0$. 
 At first, observe that if it holds that $|\mathcal{C}\cap A_1^i|\geq 1$, $|\mathcal{C}\cap A_1^j|\geq 1$, and $|\mathcal{C}\cap A_1^k|\geq 1$ for some distinct $i,j,k \in [n]$, 
    $\mathcal{C}$ induces a triangle. Hence, we have $c\leq 2$.
    The proof is completed by the following case analysis.
\begin{enumerate}
    \item \textbf{The case of $\ell \geq 3$:} 
    Suppose that $|\mathcal{C}\cap A_1| > \ell$. Then we have $c=2$. Let $i,j\in[n]$ be the indices such that $|\mathcal{C}\cap A_1^i|\geq |\mathcal{C}\cap A_1^j| > 0$.
    If $|\mathcal{C}\cap A_1^j|= 1$, then $|\mathcal{C}\cap A_1^i|\geq 3$. This is not possible since the vertices in $\mathcal{C}\cap A_1^j$
    are connected at least three vertices in $\mathcal{C}\cap A_1^i$.
    If $|\mathcal{C}\cap A_1^j|\geq 2$, then $|\mathcal{C}\cap A_1^i|\geq 2$.
    This is not possible since $\mathcal{C}$ contains a $4$-cycle in this case.
    \item \textbf{The case of $\ell = 2$:} Suppose that $|\mathcal{C}\cap A_1| > \ell=2$. Then we have $c=2$ and 
    let $i,j\in[n]$ be the indices such that $|\mathcal{C}\cap A_1^i|\geq |\mathcal{C}\cap A_1^j| > 0$. If $|\mathcal{C}\cap A_1^i|= 2$ and $|\mathcal{C}\cap A_1^j|= 2$,
    then $\mathcal{C}$ induces $C_4$. Hence, we consider $|\mathcal{C}\cap A_1^i|= 2$ and $|\mathcal{C}\cap A_1^j|= 1$. Denote $\{a_1^{i,1},a_1^{i,2}\}=\mathcal{C}\cap A_1^i, \{u\}=\mathcal{C}\cap A_1^j$. Then, $(a_1^{i,1},u),(a_1^{i,2},u)\in E\cap\{\mathcal{C}\times \mathcal{C}\}$. The other edges which incident on $a_1^{i,1},a_1^{i,2}$ are both in $A_2$ or both in $U_A$.
    \begin{enumerate}
        \item In the former case, we have that $\mathcal{C}\cap A_2=A_2^k = \{a_2^{k,1}, a_2^{k,2}\}$ for some $k\in [n]$, otherwise $\mathcal{C}$ includes an induced $5$-cycle.
        Observe that due to our construction of $G_{x,y}$, six vertices in $\mathcal{C}$ are automatically determined to $Code(A_2^k),Code(B_2^k),$ and $B_2^k$. It can be easily checked that no matter how we choose the remaining vertices, $\mathcal{C}$ does not induce a $16$-cycle.
        \item In the latter case, it is automatically determined that $\mathcal{C}$ includes $Code(A_1^i), Code(B_1^i),$ and $B_1^i$, due to our construction of $G_{x,y}$. In addition, $\mathcal{C}$ includes $B_2^k,Code(B_2^k),Code(A_2^k)$ for some $k \in [n]$. 
        Then, $\mathcal{C}$ does not induce a $16$-cycle since two vertices of $Code(A_2^k)$ do not share neighbors. 
    \end{enumerate}
\end{enumerate}
\end{proof}

\begin{lemma}\label{C8l_lemma2}
Any subset of vertices $\mathcal{C}\subseteq V$ of size $8\ell$ in $G_{x,y}$ which induces $C_{8\ell}$ contains  $\ell$ vertices in $S$, where $S\in \{A_1,A_2,B_1,B_2,U_A,U_B,L_A,L_B\}$.
\end{lemma}
\begin{proof}
For $S\subseteq V$, we denote $z(S)=|\mathcal{C}\cap S|$.
Observe that the number of edges in $E_{\mathcal{C}}$ between $A_1$ and $U_A$ is at least $z(U_A)$ since each vertex in $\mathcal{C}\cap U_A$ has at least one neighbor in $\mathcal{C}\cap A_1$. On the other hand, the number of edges in $E_{\mathcal{C}}$ between $A_1$ and $U_A$ is at most $z(A_1)$ since each vertex in $\mathcal{C}\cap A_1$ has at most one neighbor in $\mathcal{C}\cap U_A$. Hence, we have that $z(A_1) \geq z(U_A)$. Similar observation shows that $z(A_2) \geq z(L_A)$, $z(B_1) \geq z(U_B)$, and $z(B_2) \geq z(L_B)$.
Then, it holds that
\begin{align*}
    8\ell &= z(A_1)+z(A_2)+z(U_A)+z(L_A)+z(B_1)+z(B_2)+z(U_B)+z(L_B)\\
    &\leq 2(z(A_1)+z(A_2)+z(B_1)+z(B_2)).
\end{align*}
From Lemma~\ref{C8l_lemma}, we have $z(A_1)=z(A_2)=z(B_1)=z(B_2)=\ell$.
We also have $z(U_A)=z(L_A)=z(U_B)=z(L_B)=\ell$ since
\begin{align*}
    4\ell = z(U_A)+z(L_A)+z(U_B)+z(L_B).
\end{align*}
\end{proof}

\noindent\textbf{Proof of Theorem~\ref{LB_for_C8}: }
Let $\mathcal{C}\subseteq V$ be a set of $8\ell$ vertices which induces an $8\ell$-cycle in $G_{x,y}$.
From Lemma \ref{C8l_lemma} and Lemma \ref{C8l_lemma2}, $\mathcal{C}$ contains $\ell$ vertices in each $A_1^i,A_2^j,B_1^s$, and $B_2^t$ for some $i,j,s,t\in [n]$. Since $A_1^i$ must be connected to $\mathcal{C}\cap U_A$, it holds $Code(A_1^i)=\mathcal{C}\cap U_A$. Similarly, we have that $Code(A_2^j)=\mathcal{C}\cap L_A$, $Code(B_1^s)=\mathcal{C}\cap U_B$, and
$Code(B_2^t)=\mathcal{C}\cap L_B$.
It can be easily checked that $\mathcal{C}=A_1^i\cup A_2^j \cup B_1^s \cup B_2^t \cup Code(A_1^i)\cup Code(A_2^j) \cup Code(B_1^s) \cup Code(B_2^t)$ induces $C_{8\ell}$ iff $i=s, j=t,$ and $x_{i,j}=y_{i,j}=1$. \qed

\section{Proof of Lemma \ref{diamond_lemma}}\label{appB}
Let $N(a_i^j)=\{b_1,b_2,...,b_{\sqrt{n}}\}$ be the set of neighbors of $a_i^j$ in $B$. Consider $A_k$ such that $k\neq i$. 
For any $b_l\in N(a_i^j)$, just one vertex that is connected to $b_l$ is chosen uniformly at random from $A_k$. For $a\in A_k$, let $X(a)$ be the indicator variable of the event ``the pair $(a,a_i^j)$ forms a good pair''. 
Then, the expected value of $X(a)$ is
$E\left(X(a)\right) = \sqrt{n}\cdot \frac{1}{\sqrt{n}}\cdot \left( \frac{\sqrt{n}-1}{\sqrt{n}}\right)^{\sqrt{n}-1} \geq 1/e$.
Hence, the expected value of the number of vertices in $A_k$ that form good pairs with $a_i^j$ is greater than $\sqrt{n}/e$ by linearity of expectation. The expected value of the number of vertices that form good pairs with $a_i^j$ is $\sqrt{n}/e \times (\sqrt{n}-1) = \Omega(n)$. Again, by using linearity of expectation, the expected value $E(|P_A|)\geq \Omega(n)\cdot n /2=\Omega(n^2)$. This means that there exists a graph with the condition holds.
 \qed

\section{Proof of Theorem~\ref{theorem_diamond_limitation}}\label{appC}
We show the two-party communication protocol for listing diamonds by modifying the protocol for listing cliques in \cite{czumaj2020detecting}. More precisely, we show the following theorem.

\begin{theorem}
There is a two-party communication protocol in the vertex partition model for listing all diamonds that uses $\tilde{O}(\sqrt{n}|E_{cut}|)$ communication where $E_{cut}$ is a set of cut edges in the input graph.
\end{theorem}
\begin{proof}
If $|E_{cut}|\geq n^{3/2}$, Alice can send $E_A$ to Bob within $O(\sqrt{n}|E_{cut}|)$ bits of communication since it holds $\sqrt{n}|E_{cut}|\geq n^2$. Suppose that $|E_{cut}| < n^{3/2}$.
Since Alice (and Bob) can list all diamonds in which three or four vertices in Alice's side without communication, we only care about diamonds in which exactly two vertices are in Alice's side.
Let $V_A^{heavy}=\{v\in V_A: \text{deg}_{V_B}(v) > \text{deg}_{V_A}/\sqrt{n}\}$ and $V_A^{light}=V_A\backslash V_A^{heavy}$. As shown in Figure~\ref{fig:cases of diamond}, there are three possible cases.
\begin{figure*}[tbp]
    \centering
    \includegraphics[width=0.6\hsize]{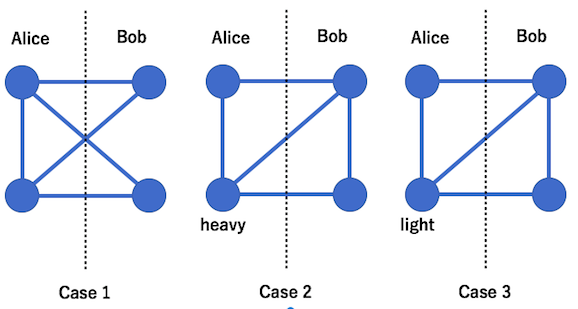}
    \caption{Three types of diamonds which have exactly two vertices in $V_A$.} 
    \label{fig:cases of diamond}
\end{figure*}

\begin{itemize}
    \item To list diamonds of case 1, we can use the protocol for listing cliques in \cite{czumaj2020detecting}. This requires $O(\sqrt{n}|E_{cut}|)$ bits.
    
    \item  To list diamonds of case 2, Alice sends edges $E_A \cap \{V_A^{heavy} \times V_A\}$ to Bob. This requires $O(\sqrt{n}|E_{cut}|)$ bits since the number of edges Alice sends to Bob is less than
    \begin{align*}
        \sum_{v\in V_A^{heavy}} \text{deg}_{V_A}(v) \leq
        \sum_{v\in V_A^{heavy}} \sqrt{n} \cdot\text{deg}_{V_B}(v)
        = \sqrt{n}|E_{cut}|.
    \end{align*}
    \item To list diamonds of case 3, for every $v\in V_A^{light}$, Bob sends edges $E_B\cap \{N_{V_B}(v) \times N_{V_B}(v)\}$ to Alice. This requires $O(\sqrt{n}|E_{cut}|)$ bits  since the number of edges Bob sends to Alice is less than
    \begin{align*}
    \sum_{v\in V_A^{light}} \left(\text{deg}_{V_B}(v)\right)^2
    \leq 
    \sum_{v\in V_A^{light}} \frac{\text{deg}_{V_A}(v)}{\sqrt{n}}\cdot\text{deg}_{V_B}(v)
    \leq 
    \sqrt{n}\sum_{v\in V_A^{light}}\text{deg}_{V_B}(v)
    \leq 
    \sqrt{n}|E_{cut}|.
    \end{align*}
\end{itemize}
\end{proof}

\setcounter{theorem}{5}
\begin{theorem}[Formal statement]\label{theo:diamond_protocol}
No family of lower bound graphs gives an $\tilde{\Omega}(n^{1/2+\varepsilon})$ lower bound of induced diamond listing for any $\varepsilon >0$.
\end{theorem}\setcounter{theorem}{18}
\begin{proof}
Exactly the same as how Theorem \ref{theorem_ck_limitation} was proved from Theorem \ref{theorem:cycle_protocol}.
\end{proof}

\section{Sublinear-round listing of induced diamonds (Theorem \ref{theorem_diamond_upper_bound})}\label{appD}
Assume that, for some edge subset $E'\subseteq E$ where $|E'|=c|E|$ for some constant $c>0$, all cliques that contain at least one edge from $E'$ are listed by a procedure $\mathcal{A}$.  We recursively apply $\mathcal{A}$ for $E\backslash E'$ since remaining cliques are the ones whose edges are in $E\backslash E'$. After $O(\log n)$ levels of recursion of $\mathcal{A}$, all cliques in the  original graph are listed since removing edges does not increase the number of cliques. Fastest (and optimal) clique listing algorithms~\cite{chang2019improved,censor2021tight} use this recursive method. On the other hand, this cannot be used for induced subgraphs since removing edges may increase the number of induced subgraphs (e.g., removing one edge from a $4$-clique creates an additional diamond).
Instead, we can use $K_4$ listing algorithm of \cite{eden2019sublinear} to list induced diamonds.
The algorithm begins by computing the decomposition of edge set, in which the edge set are decomposed into two subsets: edges that induce clusters with low mixing time and edges between clusters. 
The clusters satisfy the following:
\begin{definition}[$\delta$-cluster]For an $n$-node graph $G=(V,E)$ and a subgraph $G'=(V',E')$ of $G$, $G'$ is called a $\delta$-cluster if the following condition holds:
\begin{enumerate}
    \item Mixing time of $G'$ is $O(\mathrm{poly}\log{(n)})$,
    \item For any $v\in V'$, $\deg_{E'}(v)=\Omega(n^{\delta})$.
\end{enumerate}
\end{definition}

\begin{lemma}
[Expander Decomposition, Lemma 9 of \cite{eden2019sublinear}]\label{lemma_expander}
For a $n$-node CONGEST network $G=(V,E)$, we can find, \textit{w.h.p.}, in $\tilde{O}(n^{1-\delta})$ rounds, a decomposition of  $E$ to $E=E_m \cup E_s$satisfying the following conditions.
\begin{enumerate}
    \item $E_m$ is the union of at most $s = O(\log n)$ sets, $E_m=\displaystyle\bigcup_{i=1}^s{E_m^i}$, where each $E_m^i$ is the vertex-disjoint union of $O(n^{1-\delta})$ $\delta$-clusters, $C_i^1,\ldots,C_i^{k_i}$.
    
    The set $E_m^i$is called $i$-th level of the decomposition. We say that a node $u$ belongs to cluster $C_i^j$ if at least one of $u$'s edges is in $C_i^j$.
    \item Each level-$i$ cluster $C_i^j$has a unique identifier, which is a pair of the form$(i,x)$ where $x\in[n^{1-\delta}]$, and an unique leader node, which is some node in the cluster.
    Each node $u$ knows the identifier of all the clusters $C_i^j$ to which $u$ belongs, the leaders for those clusters, and it knows which of its edges belong to which clusters.    
    \item $E_s=\displaystyle\bigcup_{v\in V}{E_{s,v}}$, where $E_{s,v}$is a subset of edges incident to $v$ and $|E_{s,v}|\leq n^{\delta}\log n$. Each vertex $v$ knows $E_{s,v}$.
\end{enumerate}
\end{lemma}

The reason that the expander decomposition is used in the distributed subgraph detection is that $\delta$-clusters can simulate \clique style algorithms efficiently:

\begin{lemma}[Lemma 13 of \cite{eden2019sublinear}]\label{lemma_simulation}
For a constant $0<\varepsilon \leq 1$, suppose an edge set $E'$ is partitioned between the nodes of a $\delta$-cluster $C$, so that each node $u\in C$ initially knows a subset $E'_u$ of size at most $O(n^{2-\varepsilon})$. Then a simulation of $t$ rounds of the \clique algorithm on $G'=(V,E')$ can be performed in $\tilde{O}(n^{2-\delta-\varepsilon}+t\cdot n^{2-2\delta})$ rounds, with success probability $1-\frac{1}{n^2}$.
\end{lemma}
Roughly speaking, a $\delta$-cluster can simulate 1 round of a \clique style algorithm in $\tilde{O}(n^{2-2\delta})$ rounds. 

After doing the decomposition $E=E_m\cup E_s$, the high-level approach of $K_4$ listing algorithm of \cite{eden2019sublinear} is as follows:

\begin{enumerate}
    \item To list $K_4$ in $E_s$, we can use the trivial algorithm since the subgraph induced by $E_s$ is sparse.
    \item To list $K_4$ which contains at least one edge in $E_m$, each cluster $C$ gathers outside edges which are incident to a cluster node in $\tilde{o}(n)$ rounds so that the number of edges gathered by each node of $C$ is $\tilde{o}(n^2)$.
All nodes that are outside of $C$ with many neighbors in $C$ send all edges incident to them. This can be done efficiently since they have enough communication bandwidth to $C$.
Then, $C$ simulates $K_4$ listing algorithm of the \clique model.
    \item Note that each outside node $u$ with small bandwidth to $C$ have no ability to send all edges incident to it in $\tilde{o}(n)$ rounds. However, $u$ can quickly gather all cluster edges that would belong to some $K_4$ which contains $u$ since the number of cluster neighbors of $u$ is small.
\end{enumerate}

The difference is that in the case of diamonds, there are several types of diamonds which are listed in step 2 and step 3 of above algorithm (see Figure~\ref{fig:diamond_algorithm}). For a cluster $C$, we call a node $v$ is $C$-$heavy$ when $v$ does not belong to $C$ and has more than $n^{\varepsilon}$ neighbors in $C$.
The algorithm for listing induced diamonds is as follows:

\begin{enumerate}
    \item First, we run the expander decomposition of Lemma~\ref{lemma_expander} in $\tilde{O}(n^{1-\delta})$ rounds.
    \item To list diamonds in $E_s$, each node $u$ sends the edges of $E_s$ that are incident on $u$ to all neighbors. This requires $\tilde{O}(n^{\delta})$ rounds.
    \item Each $C$-$heavy$ node $v$ sends $N(v)$ to $C$ in $O(n^{1-\varepsilon})$ rounds by partitioning $N(v)$ into $|N(v)\cap C|=\Omega(n^{\varepsilon})$ subsets of size $|N(v)|/|N(v)\cap C|=O(n^{1-\varepsilon})$ and sending each subset to a different neighbor in $C$. Then, since each $C$-node $u$ receive at most $n\cdot n^{1-\varepsilon}=n^{2-\varepsilon}$ edges, we can list all induced diamonds which  contains at least one edge from some cluster $C$, and contains a $C$-$heavy$ node, in $\tilde{O}(n^{2-\delta-\varepsilon}+\sqrt{n} \cdot n^{2-2\delta})$ rounds by the algorithm of Lemma \ref{lemma_simulation}.
    \item Each $C$-$light$ node $u$ sends $N(u)\cap C$ to all its neighbors in $O(n^{\varepsilon})$ rounds. Then, each node $v$ received $N(u)\cap C$ from $u$ compute the following set locally:
    \begin{align*}
    \mathcal{L}_v^1 = &\left\{
        \{u,c_1,c_2\}\middle| u\in N(v)\text{ is $C$-$light$},c_1,c_2\in C,\right.c_1 \in N(v)\cap N(u),c_2 \in N(u)\backslash N(v)
    \left. \right\},\\
    \mathcal{L}_v^2 = &\left\{
        \{u,c_1,c_2\}\middle| u\notin N(v)\text{ is $C$-$light$},c_1,c_2\in C,\right. c_1,c_2 \in N(v)\cap N(u),\\ &(u,c_1), (u,c_2), (v,c_1),\text{ and } (v,c_2) \in E_s.
    \left. \right\},
    \end{align*}
    where $\mathcal{L}_v^1$ and $\mathcal{L}_v^2$ correspond to the rightmost and leftmost diamonds in Figure~\ref{fig:diamond_algorithm}, respectively. Note that we do not have to care about the leftmost diamond of Figure~\ref{fig:diamond_algorithm} which contains an edge from another cluster $C'$: Among the leftmost diamonds $\{u,v,c_1,c_2\}$ of Figure~\ref{fig:diamond_algorithm}, where $c_1$, $c_2$ belong to $C$ and $u,v$ are $C$-light nodes, we only need to enumerate the diamonds whose four edges $(u,c_1)$, $(u,c_2)$, $(v,c_1)$ and $(v,c_2)$ are $E_s$-edges. This is because, for instance, if $(u,c_1)$ is $C'$-edge for some cluster $C'$, then this diamond is treated by the cluster $C'$ as the middle or rightmost diamond in Figure~\ref{fig:diamond_algorithm}. Since each node sends its $E_s$-edges to all its neighbors in step 2, $v$ knows the four edges $(u,c_1)$, $(u,c_2)$, $(v,c_1)$ and $(v,c_2)$ even when $u\notin N(v)$.
    Then, $v$ construct the following list of edge queries locally:
    \begin{align*}
    \mathcal{Q}_{v,c_1} = \left\{
    \{c_1,c_2\} \middle| \exists u\in N(v):\{u,c_1,c_2\}\in \mathcal{L}^1_v
    \text{ or }\exists u\notin N(v):\{u,c_1,c_2\}\in \mathcal{L}^2_v
    \right\}.
    \end{align*}
    We have $|\mathcal{Q}_{v,c_1}|=O(n^{\varepsilon})$: if $c_1$ received $\{c_1,c_2\}\in \mathcal{Q}_{v,c_1}$, then $c_2\in N(v)\cap C$ by definition. On the other hand, since $v$ is $C$-$light$, $|N(v)\cap C|=O(n^{\varepsilon})$ holds. Each node $v$ sends $\mathcal{Q}_{v,c_1}$ to $c_1$, and $c_1$ responds with $\mathcal{Q}_{v,c_1}\cap (\{c_1\}\times N(c_1))$ in $O(n^{\varepsilon})$ rounds. Therefore, in $O(n^{\varepsilon})$ rounds, we can list all induced diamonds which contains at least one edge from some cluster $C$, and does not contain $C$-$heavy$ nodes.
\end{enumerate}

Taking parameters $\delta = 5/6$ and $\varepsilon = 1/2$ (same as in the algorithm of \cite{eden2019sublinear}), we get the following theorem.
\begin{figure*}[tbp]
    \centering
    \includegraphics[width=0.8\hsize]{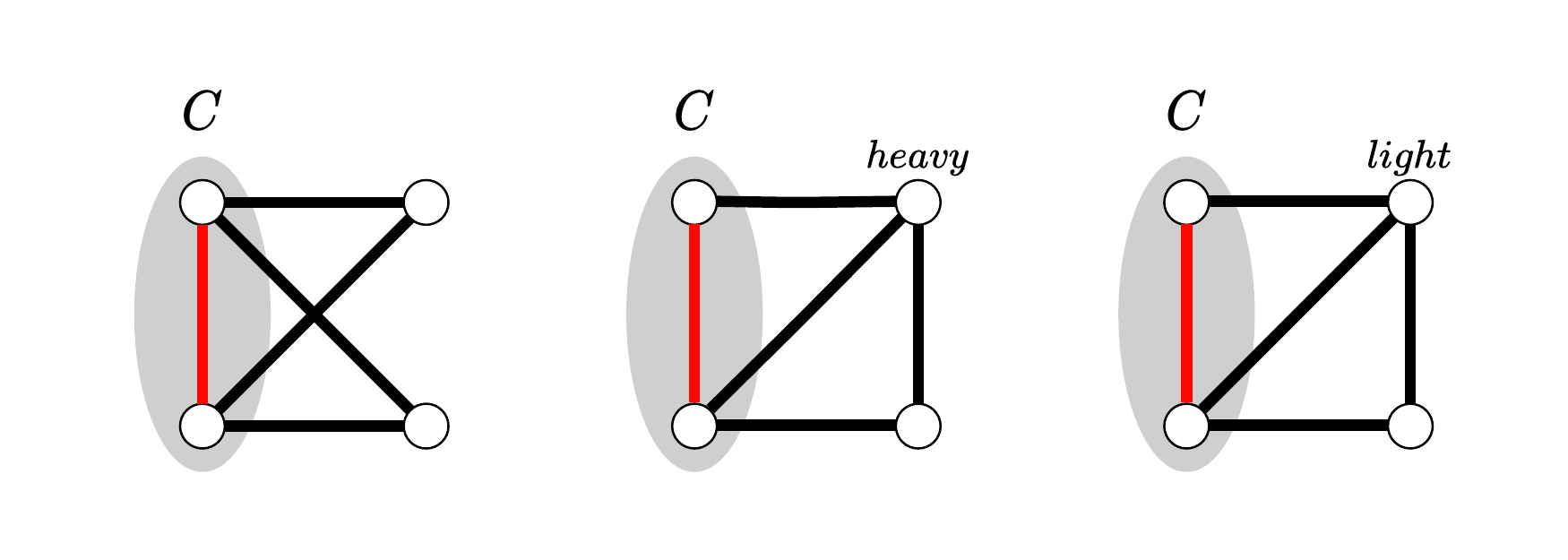}
    \caption{Three types of diamonds that contain at least one edge from $E_m$. Red edges in the figure correspond to $E_m$-edges, i.e., edges belong to some cluster $C$. Heavy node represents a $C$-$heavy$ node, i.e., a node which does not belong to $C$, but has more than $n^{\varepsilon}$ neighbors in $C$. Light node represents a $C$-$light$ node.} 
    \label{fig:diamond_algorithm}
\end{figure*}

\setcounter{theorem}{6}
\begin{theorem}\label{theo:diamond listing upper bound}
Listing all induced diamonds can be done in $\tilde{O}(n^{5/6})$ rounds with high probability in the \congest model.
\end{theorem}

\end{document}